\documentclass[11pt, letterpaper]{article}
\usepackage[utf8]{inputenc}
\usepackage[T1]{fontenc}
\usepackage[margin=1in, top=1in, bottom=1in]{geometry}

\PassOptionsToPackage{hyphens}{url}
\usepackage{hyperref}
\usepackage{url}
\usepackage{orcidlink}
\usepackage{biblatex}
\AtBeginBibliography{\small}

\RequirePackage[tt=false]{libertine}

\usepackage{titlesec}
\titleformat{\paragraph}[runin]
{\normalfont\normalsize\itshape}{\theparagraph}{1em}{}[.]

\usepackage{amsthm}
\usepackage{amsmath}
\usepackage{amsfonts}
\usepackage{mathtools}
\usepackage{cancel}
\usepackage{soul}
\usepackage{bm}
\usepackage{algorithm}

\usepackage{float}
\usepackage[justification=justified]{caption}
\usepackage{multirow}
\usepackage{array}
\newcolumntype{P}[1]{>{\centering\arraybackslash}p{#1}}

\theoremstyle{plain}
\newtheorem{theorem}{\protect\theoremname}[section]

\theoremstyle{definition}
\newtheorem{definition}[theorem]{\protect\definitionname}
\theoremstyle{definition}

\theoremstyle{remark}
\newtheorem{claim}[theorem]{\protect\claimname}
\theoremstyle{remark}
\newtheorem{example}[theorem]{Example}
\theoremstyle{remark}
\newtheorem{corollary}[theorem]{\protect\corollaryname}
\theoremstyle{remark}
\newtheorem{remark}[theorem]{Remark}
\theoremstyle{remark}
\newtheorem{lemma}[theorem]{Lemma}
\theoremstyle{remark}
\newtheorem{fact}[theorem]{Fact}

\usepackage{makecell}

\providecommand{\claimname}{Claim}
\providecommand{\corollaryname}{Corollary}
\providecommand{\definitionname}{Definition}
\providecommand{\theoremname}{Theorem}

\DeclareMathOperator{\tr}{Trace}
\DeclareMathOperator{\rk}{Rank}

\renewcommand{\i}{\bar{i}}
\renewcommand{\j}{\bar{j}}
\renewcommand{\k}{\bar{k}}

\newcommand{\szone}{{\mathtt{TA}\text{-}\mathtt{New25}}}
\newcommand{\sztwo}{{\mathtt{TA}\text{-}\mathtt{New25b}}}

\newcommand{\supplementalurl}{\url{https://www.cs.huji.ac.il/~odedsc/papers/trilinear_aggregation_algorithms_decomposed-2025-07-29.zip}}

\begin{document}

\title{Towards Faster Feasible Matrix Multiplication \\ by Trilinear Aggregation}
\author{
    Oded Schwartz \orcidlink{0000-0003-1309-5566} \\ 
    \href{mailto:odedsc@cs.huji.ac.il}{odedsc@cs.huji.ac.il} \\
    The Hebrew University of Jerusalem \\
    Jerusalem, Israel
    \and
    Eyal Zwecher \orcidlink{0009-0007-4648-904X} \\
    \href{mailto:eyal.zwecher@mail.huji.ac.il}{eyal.zwecher@mail.huji.ac.il} \\
    The Hebrew University of Jerusalem \\
    Jerusalem, Israel
}
\date{}
\maketitle

\begin{abstract}
Matrix multiplication is a fundamental kernel in high performance computing. Many algorithms for fast matrix multiplication can only be applied to enormous matrices ($n>10^{100}$) and thus cannot be used in practice. Of all algorithms applicable to feasible input, Pan's $O(n^{2.773372})$ algorithm (1982) is asymptotically the fastest. We obtain an $O(n^{2.773203})$ algorithm applicable to the same input sizes as Pan's algorithm. This algorithm is the fastest matrix multiplication algorithm with base case smaller than $1000$. Further, our method obtains the best asymptotic complexity for many small base cases, starting at $n_0=28$. We also obtain better exponents for larger base cases.
To construct our algorithm, we use the trilinear aggregation method. We find parts of the algorithms that are equivalent to matrix multiplication with smaller base case, and use the de Groote equivalence to replace these parts in a way that allows further optimization of our algorithms. Finally, we improve the additive complexity of our algorithms by finding a sparse decomposition and reducing the leading coefficient. These mark a fundamental step towards outperforming existing fast matrix multiplication algorithms in practice.

\end{abstract}

\section{Introduction}
Matrix multiplication is a fundamental operation in computer science, and is used in many applications, including scientific computations, and machine learning. Hence, optimizing matrix multiplication is essential. In 1969, Strassen obtained the first sub-cubic time matrix multiplication algorithm, with complexity of $O(n^{2.807355})$ \cite{strassen1969}. It took nearly a decade to obtain a faster algorithm - Pan's trilinear aggregation algorithm \cite{pan78}. Much of later research uses techniques which result in algorithms applicable only to matrices of huge dimensions ($n>10^{100}$), due to huge base case sizes, as well as very large constants hidden in the $O$-notation, and are thus impractical (cf. \cite{alman2024asymmetryyieldsfastermatrix, cohn2003group, coppersmith1987matrix, davie2013improved, le2014powers, williams2012multiplying, williams2023newboundsmatrixmultiplication, strassen1986asymptotic, bini_first_apa_algorithm, partial_and_total, duan2023asymetric_hashing, Alman2024Refined}). Nonetheless, many studies produce matrix multiplication algorithms applicable to \emph{feasible}\footnote{Feasible algorithms are algorithms that can run on real world hardware and perform multiplication of matrices that are needed for practical applications. Thus, the soft definition of \emph{feasible algorithms} changes based on technological advances. Nevertheless, all known matrix multiplication algorithms we are aware of have either relatively small base cases $n_0<10^4$ or enormous base cases $n_0>10^{100}$.} size (cf. \cite{lipshitz2012communication, laderman1992practical, smirnov2013bilinear, drevet2011optimization, pan82}) with several new recent ones (cf. \cite{flipgraph, moosbauer2025flip, local_search, kaporin4_4_4_48, dumas2025noncommutative, alphaevolve, google_algorithms}). Yet, Pan's $O\left(n^{2.773372}\right)$ algorithm (1982) is asymptotically the fastest among those.

\subsection{Previous Work}

\paragraph{Techniques resulting in huge dimensions}
The minimal $\omega$ such that two matrices of size $n\times n$ can be multiplied in time $O\left(n^{\omega +\varepsilon}\right)$ for every $\varepsilon>0$ is called the matrix multiplication exponent \cite{partial_and_total}. Finding $\omega$ and whether $\omega=2$ is a fundamental open problem in theoretical computer science \cite{von1988algebraic, legall_bound_of_coppersmith}. Bini et al. \cite{bini_first_apa_algorithm} obtained an $O\left(n^{2.779886}\right)$ approximate algorithm for matrix multiplication, which is applicable to small input size. Bini later demonstrated how to convert it to an exact algorithm \cite{bini1980relations}. However, the exactification process creates an algorithm with an enormous base case. Sch\"{o}nhage \cite{partial_and_total} obtained an $O\left(n^{2.521813}\right)$ algorithm by introducing the $\tau$-theorem. This theorem converts disjoint and partial matrix multiplication algorithms into an algorithm that multiplies one pair of matrices. Again, the resulting algorithm has a huge base case. Strassen \cite{strassen1986asymptotic} later obtained an $O\left(n^{2.478496}\right)$ algorithm using the laser method. This method requires exponentiation of a tensor with a low rank, and thus results in huge base cases. In 1987, Coppersmith and Winograd \cite{coppersmith1987matrix} used this method to obtain an $O\left(n^{2.387190}\right)$ algorithm by using the laser method on a different tensor. The Coppersmith-Winograd tensor cannot yield complexity better than $O\left(n^{2.3078}\right)$ \cite{legall_bound_of_coppersmith}. Nevertheless, almost all upper bounds on $\omega$ since 1987 utilize the Coppersmith-Winograd method (cf. \cite{coppersmith1987matrix, williams2012multiplying, le2014powers, duan2023asymetric_hashing, williams2023newboundsmatrixmultiplication, Alman2024Refined, alman2024asymmetryyieldsfastermatrix}). At the time of writing, the best upper bound on $\omega$ is $\omega \leq 2.371339$ by Alman et al. \cite{alman2024asymmetryyieldsfastermatrix}.

\paragraph{Feasible Algorithms}

Strassen's recursive algorithm uses a base case that multiplies a $2\times 2$ matrix by a $2\times 2$ matrix, with $7$ multiplications. Namely, it is a $\left<2,2,2;7\right>$-algorithm\footnote{An $\left<m_0, n_0, p_0; t\right>$-algorithm is an algorithm for multiplying an $m_0\times n_0$ matrix by an $n_0 \times p_0$ matrix using $t$ multiplications. See Section \ref{sec: preliminaries}, Definition \ref{def: n0, m0, r0; t algorithm}. This means that $\rk(\left<n_0, m_0, p_0\right>)\leq t$, where $\left<n_0, m_0, p_0\right>$ is the matrix multiplication tensor, see Definition \ref{def: mm tensor} and Claim \ref{claim: tensor rank and mm}.}. Brent \cite{brent_equations} describes a set of trilinear equations on the parameters of a bilinear algorithm that are sufficient and necessary for the algorithm to compute matrix multiplication. Laderman \cite{laderman333-23} found a solution to these equations and obtained a $\left<3,3,3;23\right>$-algorithm with complexity of $O(n^{2.85405})$, which is not faster than Strassen's algorithm.
Nearly a decade after Strassen obtained the first non-trivial matrix multiplication exponent, Pan \cite{pan78} improved upon Strassen's bound, by the trilinear aggregation technique obtaining an $O\left(n^{2.795123}\right)$ algorithm. He later improved that to $O\left(n^{2.780142}\right)$ \cite{pan80}, and to $O\left(n^{2.773372}\right)$ \cite{pan82} by introducing implicit canceling. His algorithms are applicable to small input sizes with the best algorithm having a base case $n_0=44$. Four decades after they were published, they still outperform all other feasible algorithms starting at $n_0=30$. Johnson and McLoughlin \cite{johnson_mcloughlin_33_23} found further $\left<3,3,3;23\right>$-algorithms. Laderman et al. \cite{laderman1992practical} and Drevet et al. \cite{drevet2011optimization} found new algorithms based on Pan's trilinear aggregation. They outperform Pan's algorithm for base cases $n_0 \le 28$, though none are asymptotically faster than Pan's $O(n^{2.773382})$ algorithm \cite{pan82}. Drevet et al. \cite{drevet2011optimization} also describe techniques for combining existing matrix multiplication algorithms into new ones.
Smirnov \cite{smirnov2013bilinear} discovered a  $\left<3,3,6;40\right>$-algorithm with complexity of $O\left(n^{2.77430}\right)$\footnote{Algorithms for rectangular matrix multiplication can be converted into square matrix multiplication algorithms using symmetrization \cite{hopcroft1973duality}. See Claim \ref{claim: symmetrization}. For ease of comparison we provide the exponent of the square variant.} by relaxing Brent's equations and then solving iteratively using the least squares method. This algorithm is asymptotically faster than Strassen's algorithm. Tichavsk\`{y} and Kov\'{a}\v{c} obtained a different $\left<3,3,6;40\right>$-algorithm \cite{tichavsky_kovac_336_40}. Ballard and Benson \cite{benson2015framework} obtained new algorithms using computer aided search similar to \cite{smirnov2013bilinear, johnson_mcloughlin_33_23}, though none are better than Strassen's algorithm. Heule et al. \cite{local_search, new_ways_to_multiply_33} used SAT-based methods and discovered thousands of new $\left<3,3,3;23\right>$-algorithms. Fawzi et al. \cite{google_algorithms} found new algorithms with small base cases, such as a $\left<3,4,5;47\right>$-algorithm and $\left<4,5,5;76\right>$-algorithm using the reinforcement learning model AlphaTensor. These algorithms have complexities of $O\left(n^{2.821073}\right)$ and $O\left(n^{2.821221}\right)$ respectively. Kauers and Moosbauer \cite{flipgraph} obtained new algorithms, including a $\left<4,4,5;62\right>$-algorithm with complexity of $O\left(n^{2.825498}\right)$, using the flip graph method. This method was further used to obtain new algorithms \cite{moosbauer2025flip, arai2024adaptive, kauers_short_communication}. Kaporin \cite{kaporin4_4_4_48} obtained a $\left<4,4,4;48\right>$-algorithm with complexity of $O\left(n^{2.792482}\right)$ by solving a nonlinear least squares problem. This exponent is better than Strassen's algorithm. This result was matched by Novikov et al. \cite{alphaevolve}, who obtained a different $\left<4,4,4;48\right>$-algorithm over the complex numbers. Dumas et al. \cite{dumas2025noncommutative} used the complex $\left<4,4,4;48\right>$-algorithm by Novikov et al. \cite{alphaevolve} to obtain a $\left<4,4,4;48\right>$-algorithm with non-complex coefficients.

To-date, only a few feasible algorithms have better exponent than Strassen's $O\left(n^{\log_2 7}\right)$ algorithm. These include Pan's algorithms \cite{pan78, pan80, pan82}, Laderman et al.'s algorithms \cite{laderman1992practical}, Drevet et al.'s algorithms \cite{drevet2011optimization}, the $\left<3,3,6;40\right>$-algorithms by Smirnov \cite{smirnov2013bilinear} and by Tichavsk\`{y} and Kov\'{a}\v{c} \cite{tichavsky_kovac_336_40}, and the $\left<4,4,4;48\right>$-algorithm by Dumas et al. \cite{dumas2025noncommutative}. Some of the feasible algorithms are applicable only to certain fields, such as a $\left<4,4,4;47\right>$-algorithm \cite{google_algorithms} and a $\left<4,4,5;60\right>$-algorithm \cite{flipgraph} that work over $\mathbb{Z}_2$.


\paragraph{Trilinear Aggregation}
Pan \cite{pan78, pan80} introduces the techniques of trilinear aggregation, uniting, and canceling (see Section \ref{sec: preliminaries}). These techniques result in matrix multiplication algorithms with better asymptotic complexity than Strassen's algorithm for many even base cases starting at $n_0\geq 20$. He later introduced implicit canceling \cite{pan82}, which further improved the complexity of the algorithms. Using these techniques he obtained a $\left<44,44,44;36133\right>$-algorithm. This algorithm has aymptotic complexity of $O(n^{2.773372})$ and is thus the fastest feasible algorithm to-date. Laderman et al. \cite{laderman1992practical} (see correction in \cite{respondek2024correction}) present new schemes for matrix multiplication based on \cite{pan80, pan82}. They describe an improved cancellation technique compared to \cite{pan80}, though they do not obtain algorithms that are asymptotically faster than \cite{pan82}. Drevet et al. \cite{drevet2011optimization} improve trilinear aggregation for smaller base cases. They also describe a variation of their algorithms for odd dimension base cases. Particularly, their algorithms have the best known exponents for base cases $22 \leq n_0 \leq 28$. Thus, for small base cases starting at $n_0 = 22$, the algorithms with the best exponents are trilinear aggregation based algorithms, using the methods of \cite{pan82} or \cite{drevet2011optimization}. For $n_0<22$, the best exponent has been obtained using other methods, sometimes in combination with trilinear aggregation (cf. \cite{strassen1969, laderman333-23, moosbauer2025flip, kaporin4_4_4_48, smirnov2013bilinear, drevet2011optimization, laderman1992practical, hopcroft1971minimizing}).

\paragraph{Leading Coefficient}
Of the fast matrix multiplication algorithms with small base cases that are known today, only few are used in practice. One reason is that many have large constants hidden in the $O$-notation. For example, Pan's $O\left(n^{2.773372}\right)$ algorithm \cite{pan82} has the best exponent. However, its leading coefficient is $737.4$. Thus, it can outperform Strassen's algorithm only for enormous input sizes $n>10^{60}$ \cite{hs23}. Even when the coefficients are not too large, the smaller they are the more practical the algorithm becomes. Hence, many studies deal with reducing these coefficients (cf. \cite{winogradadditivecomplexity, probert1976additive, bshouty1995additive, bodrato2010strassen, cenk2017arithmetic, ks, bs19, nathan2020sparsifying, martensson_number_of_the_beast, hs23}).

Winograd \cite{winogradadditivecomplexity} reduces the leading coefficient of Strassen's algorithm from $7$ to $6$. Probert \cite{probert1976additive} and Bshouty \cite{bshouty1995additive} show this to be optimal under the implicit assumptions that the algorithm has a uniform recursive structure, and that the matrices are given in the standard basis. Bodrato \cite{bodrato2010strassen} reduces the leading coefficient of Strassen's algorithm to $5$ when the multiplication is done as part of repeated squaring or chain multiplication. Cenk and Hasan \cite{cenk2017arithmetic} reduce the leading coefficient to $5$ while increasing communication costs and memory footprint by creating a non-uniform recursive structure. Karstadt and Schwartz \cite{ks} reduce the leading coefficient to $5$ at the cost of a low overhead, while maintaining low communication cost and memory footprint. To this end, they apply fast transformations to the matrices and perform the multiplication in an alternative basis. Beniamini and Schwartz \cite{bs19} generalize the alternative basis method into the sparse decomposition method, further reducing the leading coefficient, in some cases down to $2$, while increasing communication costs. Beniamini et al. \cite{nathan2020sparsifying} improve the leading coefficient of several algorithms by using a search heuristic for sparse decomposition. Hadas and Schwartz \cite{hs23} use the sparse decomposition technique to reduce the leading coefficient of Pan's $O\left(n^{2.773372}\right)$ algorithm \cite{pan82} from $737.4$ to $8.08$ and of the $O\left(n^{2.780142}\right)$ algorithm \cite{pan80} from $23.18$ to $2$, while increasing communication costs. M{\aa}rtensson and Wagner \cite{martensson_number_of_the_beast} reduce the leading coefficient of many algorithms by reusing computations without introducing an overhead, and improve the leading coefficients of \cite{williams2023newboundsmatrixmultiplication, laderman333-23, smirnov2013bilinear}.

\subsection{Our Contribution}
\begin{table}[t]
\centering
\begin{tabular}{|c||P{25mm}|P{30mm}||c|c|}
    \hline
     \multirow{2}{*}{Base case $n_0$} & \multicolumn{2}{c||}{Tensor Rank Bound $M(n_0)$} & \multicolumn{2}{c|}{Exponent $\omega_0$} \\
     \cline{2-5}
     & Previous & Here & Previous & Here \\
     \hline\hline
\hline
$28$ & $10556$ \cite{drevet2011optimization} & $10550$ & $2.780277$ & $2.780106$\\
\hline
$30$ & $12704$ \cite{pan82} & $12688$ & $2.778337$ & $2.777967$\\
\hline
$32$ & $15113$ \cite{pan82} & $15096$ & $2.776701$ & $2.776376$\\
\hline
$34$ & $17808$ \cite{pan82} & $17790$ & $2.775498$ & $2.775211$\\
\hline
$36$ & $20805$ \cite{pan82} & $20786$ & $2.774633$ & $2.774378$\\
\hline
$38$ & $24120$ \cite{pan82} & $24100$ & $2.774037$ & $2.773809$\\
\hline
$40$ & $27769$ \cite{pan82} & $27748$ & $2.773655$ & $2.77345$\\
\hline
$42$ & $31768$ \cite{pan82} & $31746$ & $2.773444$ & $2.773258$\\
\Xhline{2\arrayrulewidth}
$\mathbf{44}$ & $\mathbf{36133}$ \cite{pan82} & $\mathbf{36110}$ & $\mathbf{2.773372}$ & $\mathbf{2.773203}$\\
\Xhline{2\arrayrulewidth}
$46$ & $40880$ \cite{pan82} & $40856$ & $2.773412$ & $2.773258$\\
\hline
$48$ & $46025$ \cite{pan82} & $46000$ & $2.773543$ & $2.773403$\\
\hline
$50$ & $51584$ \cite{pan82} & $51558$ & $2.773749$ & $2.77362$\\
\hline
$60$ & $86149$ \cite{pan82} & $86118$ & $2.775496$ & $2.775408$\\
\hline
\end{tabular}

\caption{Comparing number of multiplications (bound on the tensor rank, see Claim \ref{claim: tensor rank and mm}) and resulting exponent $\omega_0$ for various base case dimensions $n_0$. Using the algorithms with base case $n_0$ recursively results in a matrix multiplication algorithm with complexity of $O\left(n^{\omega_0}\right)$ where $\omega_0 = \log_{n_0}M(n_0)$.
The exponent is minimal for base case of dimension $n_0=44$. This is the best exponent of all algorithms applicable to input matrices with feasible dimensions.}
\label{tab: comparing exponents of small algorithm}

\end{table}

\begin{table}[t]
\centering
\begin{tabular}{|l@{$\null=\null$}l||c|c|c||c|c|c|}
    \hline
     \multicolumn{2}{|c||}{\multirow{3}{*}{Base case $n_0$}} & \multicolumn{3}{c||}{Tensor Rank Bound $M(n)$} & \multicolumn{3}{c|}{Exponent $\omega_0$} \\
     \cline{3-8}
      \multicolumn{2}{|c||}{} & \multicolumn{2}{c|}{Two recursive steps} & \multirow{2}{*}{$\sztwo$} & \multicolumn{2}{c|}{Two recursive steps} & \multirow{2}{*}{$\sztwo$} \\
      \cline{3-4}\cline{6-7}
      \multicolumn{2}{|c||}{} & \cite{pan82} &  $\szone$ && \cite{pan82} & $\szone$ & \\
     \hline\hline
$28^{2}$ & $784$ &$111619225$ & $111302500$ & $111258400$ & $2.780533$ & $2.780106$ & $2.780047$\\
\hline 
$30^{2}$ & $900$ &$161391616$ & $160985344$ & $160927744$ & $2.778337$ & $2.777967$ & $2.777914$\\
\hline 
$32^{2}$ & $1024$ &$228402769$ & $227889216$ & $227815232$ & $2.776701$ & $2.776376$ & $2.776329$\\
\hline 
$34^{2}$ & $1156$ &$317124864$ & $316484100$ & $316390464$ & $2.775498$ & $2.775211$ & $2.775169$\\
\hline 
$36^{2}$ & $1296$ &$432848025$ & $432057796$ & $431940832$ & $2.774633$ & $2.774378$ & $2.774340$\\
\hline 
$38^{2}$ & $1444$ &$581774400$ & $580810000$ & $580665600$ & $2.774037$ & $2.773809$ & $2.773775$\\
\hline 
$40^{2}$ & $1600$ &$771117361$ & $769951504$ & $769775104$ & $2.773655$ & $2.773450$ & $2.773418$\\
\hline 
$42^{2}$ & $1764$ &$1009205824$ & $1007808516$ & $1007595072$ & $2.773444$ & $2.773258$ & $2.773230$\\
\Xhline{2\arrayrulewidth}$\mathbf{44^{2}}$ & $\mathbf{1936}$ &$\mathbf{1305593689}$ & $\mathbf{1303932100}$ & $\mathbf{1303676064}$ & $\mathbf{2.773372}$ & $\mathbf{2.773203}$ & $\mathbf{2.773177}$\\
\Xhline{2\arrayrulewidth}
$46^{2}$ & $2116$ &$1671174400$ & $1669212736$ & $1668908032$ & $2.773412$ & $2.773258$ & $2.773234$\\
\hline 
$48^{2}$ & $2304$ &$2118300625$ & $2116000000$ & $2115640000$ & $2.773543$ & $2.773403$ & $2.773381$\\
\hline 
$50^{2}$ & $2500$ &$2660909056$ & $2658227364$ & $2657804864$ & $2.773749$ & $2.773620$ & $2.773600$\\
\hline 
$60^{2}$ & $3600$ &$7421650201$ & $7416309924$ & $7415445024$ & $2.775496$ & $2.775408$ & $2.775394$\\
\hline 
\end{tabular}

\caption{Comparing number of multiplications (bound on tensor rank, see Claim \ref{claim: tensor rank and mm}) and resulting exponent $\omega_0$ for various base case dimensions $n_0$. The results of \cite{pan82}, and $ALG_1$ are based on two recursive steps of the base algorithms. The complexity of the algorithms is $O\left(n^{\omega_0}\right)$ where $\omega_0 = \log_{n_0}M(n_0)$.
The exponent of $\sztwo$ is minimal for base case of dimension $n_0=44^2$.}

\label{tab: comparing exponents of all algorithms}

\end{table}
We obtain two new families of fast matrix multiplication algorithms that are applicable to feasible sizes (see Tables \ref{tab: comparing exponents of small algorithm} and \ref{tab: comparing exponents of all algorithms}). The first family of algorithms (Table \ref{tab: comparing exponents of small algorithm}) contains a $\left<44,44,44;36110\right>$-algorithm with complexity of $O\left(n^{2.773203}\right)$. This implies that $\rk\left(\left<44,44,44\right>\right)\leq 36110$ (see Claim \ref{claim: tensor rank and mm}), where $\left<44,44,44\right>$ is the tensor representing $44\times 44$ matrix multiplication.\footnote{For a formal definition of this tensor, refer to Definition \ref{def: mm tensor}.} These algorithms have exponents better than Pan's \cite{pan82} and are applicable to the same input dimensions. They improve the exponents for many feasible even base cases of size $n_0 \geq 28$. The second family of algorithms further reduces the exponents, at the cost of squaring the base case size (see Table \ref{tab: comparing exponents of all algorithms}). The best algorithm of this family is a $\left<1936, 1936, 1936; 1303676064\right>$-algorithm with complexity of $O\left(n^{2.7731768}\right)$. 

Both new families of algorithms are based on the trilinear aggregation method \cite{pan1972schemes, pan80, pan82} and the implicit canceling method \cite{pan82}. In both families, we find sets of rows which are isomorphic to smaller matrix multiplication algorithms. We replace these rows with other algorithms to obtain improved algorithms. To obtain the first family of algorithms, we utilize algorithms from the de Groote \cite{DEGROOTE1978PT2} equivalence class of Strassen's $\left<2,2,2;7\right>$-algorithm to create duplicate rows in the encoding and decoding matrices. We then unite these rows to reduce multiplications. To obtain the second family of algorithms we replace sub tensors that are isomorphic to $\left<4,4,4;49\right>$-algorithms with a $\left<4,4,4;48\right>$-algorithm \cite{kaporin4_4_4_48, alphaevolve, dumas2025noncommutative}.

Combined with leading coefficient reduction from $736.3$ to $8.17$, this work is an important step towards accelerating matrix multiplication in practice.

\subsection{Organization}
Section \ref{sec: preliminaries} provides preliminaries regarding matrix multiplication, tensor operations, and the methods of trilinear aggregations and implicit canceling. Section \ref{sec: improving pans algorithms} describes our first family of algorithms. Section \ref{sec: using kaporins algorithm} presents the second family of algorithms. In Section \ref{sec: discussion} we compare our algorithms with existing ones, and discuss open problems.
The encoding and decoding matrices (see Definition \ref{def: encoding decoding matrices}) are provided as supplemental material to this work \footnote{\supplementalurl}.

\section{Preliminaries}
\label{sec: preliminaries}
\subsection{Fast Matrix Multiplication}
\begin{definition}
    Let $r\in\mathbb{N}$. We define $\left[r\right] = \left\{0,\dots,r-1\right\}$.
\end{definition}

\begin{definition}
    Let $U\in\mathbb{F}^{m\times n}$. We denote the $i$-th row of $U$ by $[U]_{i}$.
\end{definition}

\begin{definition}[Notation 2.1 in \cite{bs19}]
    \label{def: vectorization}
    Let $A\in\mathbb{F}^{m\times n}$ be a matrix, where $m=m_0^{\ell}$ and $n=n_0^{\ell}$.
    Let $A_{i,j}$ be the $(i,j)$-th block of size $\frac{m}{m_0}\times\frac{n}{n_0}$. The \textit{vectorized form} of $A$ is recursively defined as $$\vec{A}=\left(\vec{A}_{1,1}\dots\vec{A}_{1,n_0}\dots\vec{A}_{m_0,1}\dots\vec{A}_{m_0,n_0}\right)^T$$
\end{definition}

\begin{claim} [Fact 2.7 in \cite{bs19}]
\label{claim: encoding and decoding matrices}
Let $n,m,k\in\mathbb{N}$. Let $f(x,y):(\mathbb{F}^{n}\times \mathbb{F}^{m})\to\mathbb{F}^{k}$ be a bilinear algorithm that performs $t$ multiplications. There exist three matrices $U\in\mathbb{F}^{t\times n}, V\in\mathbb{F}^{t\times m},W\in\mathbb{F}^{t\times k}$ such that for all $x\in\mathbb{F}^{n}, y\in\mathbb{F}^{m}$,
$f(x,y)=W^{T}((U\cdot \vec{x}) \odot (V \cdot \vec{y}))$ where $\odot$ is the Hadamard (element-wise) product.
\end{claim}

\begin{definition}
    \label{def: encoding decoding matrices}
    A recursive fast matrix multiplication algorithm computing the matrix product $C=AB$ is represented by a triplet of matrices $\left< U,V,W\right>$ such that $\vec{C^T}=W^T\left((U\cdot\vec{A})\odot(V\cdot\vec{B})\right)$ for all matrices $A\in\mathbb{F}^{n_0\times m_0}, B\in\mathbb{F}^{m_0\times p_0}$. $U$ and $V$ are the \emph{encoding matrices} while $W$ is the \emph{decoding matrix}.
\end{definition}

\begin{remark}
    Notice that the algorithm computes $C^T$ instead of $C$. This allows an elegant relation between the bilinear and trilinear forms of the algorithm, as described in Fact \ref{fact: trilinear form of mm algorithm}.
\end{remark}

\begin{definition}
        \label{def: n0, m0, r0; t algorithm}
    An $\left<n_0, m_0, p_0; t\right>$-algorithm is an algorithm for multiplying an $n_0\times m_0$ matrix by an $m_0 \times p_0$ matrix using $t$ multiplications.
\end{definition}

\begin{definition}
        Let $a, b, c\in\mathbb{F}^t$. The \emph{triple inner product} of $a, b, c$ is  defined as
    $\left<a,b,c\right> = \sum_{i\in[t]}a_i b_i c_i$.
\end{definition}

\begin{claim}[\cite{brent_equations}]
    Let $U\in\mathbb{F}^{t\times m_0\cdot n_0}, V\in\mathbb{F}^{t\times n_0\cdot p_0}, W\in\mathbb{F}^{t\times p_0 \cdot m_0}$. 
    Then $\left<U,V,W\right>$ are the encoding and decoding matrices of an $\left<m_0, n_0, p_0;t\right>$-algorithm if and only if for all $i_1,i_2\in[n_0],j_1,j_2\in[p_0], k_1,k_2\in[m_0]$,
    $$\left<u_{n_0\cdot k_2 + i_1}, v_{p_0\cdot i_2 + j_1}, w_{m_0 \cdot j_2 + k_1}\right> = \delta_{i_1,i_2}\delta_{j_1,j_2}\delta_{k_1,k_2}$$
    where $u_r$ is the $r$-th column of $U$, and similarly for $v_r, w_r$.
\end{claim}

\begin{fact}
    Let $A\in\mathbb{F}^{m_0\times n_0}, B\in\mathbb{F}^{n_0 \times p_0},C\in\mathbb{F}^{p_0\times m_0}$. Then
    $\tr(ABC)=\sum_{i,j,k} A_{ij}B_{jk}C_{ki}$.
\end{fact}

\begin{fact} [see \cite{pan78}]
\label{fact: trilinear form of mm algorithm}
Let $U\in\mathbb{F}^{t\times m_0\cdot n_0}, V\in\mathbb{F}^{t\times n_0\cdot p_0}, W\in\mathbb{F}^{t\times p_0 \cdot m_0}$. Then $\left<U,V,W\right>$ is an $\left<m_0,n_0,p_0;t\right>$-algorithm if and only if for all $A\in\mathbb{F}^{m_0\times n_0}, B\in \mathbb{F}^{n_0 \times p_0}, C\in\mathbb{F}^{p_0 \times m_0}$, $$\tr(ABC)=\left<U\cdot\vec{A}, V\cdot\vec{B}, W\cdot\vec{C}\right>$$
\end{fact}

\begin{remark}
    By Fact \ref{fact: trilinear form of mm algorithm}, an $\left<m_0, n_0, p_0; t\right>$-algorithm is equivalent to an algorithm computing the trace of a product of three matrices $\tr(ABC)$ using $t$ multiplications where $A\in\mathbb{F}^{m_0\times n_0},B\in\mathbb{F}^{n_0\times p_0}, C\in\mathbb{F}^{p_0 \times m_0}$.
\end{remark}

\begin{fact}
    \label{fact:trace is cyclic}
    Let $A\in\mathbb{F}^{n\times m}, B\in\mathbb{F}^{m\times n}$.
    Then $\tr(AB)=\tr(BA)$.
\end{fact}

\begin{corollary} [\cite{hopcroft1973duality}]
\label{cor: cyclicness of UVW}
    Let $\left<U,V,W\right>$ be an $\left<m_0,n_0,p_0;t\right>$-algorithm. Then the algorithm obtained by cyclic rotation $\left<V,W,U\right>$ is an $\left<n_0, p_0, m_0;t\right>$-algorithm.
\end{corollary}

\begin{claim}
    \label{claim: composition of algorithms}
    If there exist an $\left<m_0, n_0, p_0; t\right>$-algorithm and an $\left<m_0', n_0', p_0'; t'\right>$-algorithm, then there exists an $\left<m_0m_0',n_0n_0', p_0p_0', tt'\right>$-algorithm.
\end{claim}

\begin{claim}[\cite{hopcroft1973duality}]
    \label{claim: symmetrization}
    If there exists an $\left<m_0, n_0, p_0; t\right>$-algorithm, then there exists an $\left<r_0, r_0, r_0;t^3\right>$-algorithm where $r_0=m_0n_0p_0$.
\end{claim}

\begin{definition}
    \label{def: mm tensor}
    The \emph{matrix multiplication tensor} $\left<m_0, n_0, p_0\right>$ is the tensor $T\in\mathbb{R}^{m_0n_0\times n_0p_0\times p_0m_0}$ such that $T_{r,s,t}=1$ if and only if $r=m_0\cdot k + i, s=n_0 \cdot i + j, t=p_0 \cdot j + k$ for some $i\in[m_0], j\in[n_0], k\in[p_0]$ and $0$ otherwise.
\end{definition}

\begin{claim}[\cite{Strassen1973}]
    \label{claim: tensor rank and mm}
    $\rk(\left<m_0,n_0,p_0\right>) \leq t$ if and only if there exists an $\left<m_0, n_0, p_0;t\right>$-algorithm.
\end{claim}

\begin{definition}[\cite{pan80}]
    Let $\left<U,V,W\right>$ be a matrix multiplication algorithm.
    We say that the $i$-th and $j$-th rows are \emph{kin} if the $i$-th and $j$-th rows are equal to one another in at least two of the three matrices, namely, at least two of the following three equalities hold: $[U]_{i}=[U]_j, [V]_{i}=[V]_j, [W]_{i}=[W]_j$.
\end{definition}

\begin{lemma}[\cite{pan80}]
    \label{lemma: kin terms reduce multiplications}
    If $\left<U,V,W\right>$ is an $\left<m_0,n_0,p_0;t\right>$-algorithm that has $s$ disjoint pairs of kin rows, then there exists an $\left<m_0,n_0,p_0;t-s\right>$-algorithm.
\end{lemma}

\begin{lemma}[\cite{pan80}]
    If $\left<U,V,W\right>$ is a an $\left<m_0, n_0, p_0; t\right>$-algorithm that has a pair of kin rows, then there exists a $\left<m_0, n_0, p_0, t-1\right>$-algorithm.
\end{lemma}

\begin{proof}
    Let the $i$-th and $j$-th rows be kin where $i<j$. W.L.O.G. (see Corollary \ref{cor: cyclicness of UVW}) $[U]_i=[U]_j$, and $[V_i]=[V_j]$. The algorithm $\left<U',V',W'\right>$ obtained from $\left<U,V,W\right>$ by removing the $j$-th row of $U,V$, and $W$ and setting $[W']_i=[W]_i+[W]_j$ is an $\left<m_0, n_0, p_0; t-1\right>$-algorithm.
\end{proof}

\begin{proof} [Proof of Lemma \ref{lemma: kin terms reduce multiplications}]
    Apply the previous lemma to $\left<U,V,W\right>$ $s$ times. Each time, we reduce the number of multiplications by $1$. Since the kin rows are disjoint, at the $i$-th iterations we have $s-i$ disjoint pairs of kin rows. Thus, applying the corollary $s$ times yields an algorithm that requires $t-s$ multiplications.
\end{proof}

We next introduce a theorem that allows creating $\left<2,2,2;7\right>$-algorithms with first rows of our choosing. Combining these algorithms into Pan's algorithms, we generate kin rows. These allow reducing the number of multiplications (see Section \ref{sec: improving pans algorithms}).

\begin{theorem}
    \label{theorem: we have 2227 alg with U0=KU and V0=KV}
    Let $K_U,K_V\in\mathbb{F}^{2\times 2}$ be invertible matrices. Then there exists a triplet $\left<U,V,W\right>$ representing a $\left<2,2,2;7\right>$-algorithm such that $[U]_0=\overrightarrow{K_U}$ and $[V]_0=\overrightarrow{K_V}$.
\end{theorem}

We sketch a proof of this theorem. de Groote \cite{DEGROOTE1978pt1, DEGROOTE1978PT2} introduces equivalence classes for matrix multiplication, and shows that all $\left<2,2,2;7\right>$-algorithms are equivalent. We search the equivalence class of Strassen's $\left<2,2,2;7\right>$-algorithm \cite{strassen1969} and find the desired algorithm. The full proof is provided in Appendix \ref{app: 222-7 algorithms}.

\subsection{Trilinear Aggregation Based Algorithms}
We next briefly describe Pan's technique of trilinear aggregation \cite{pan1972schemes, pan80, pan82, laderman1992practical}. Following Pan's conventions \cite{pan78, pan80, pan82, pan1972schemes}, we use the notation of the trilinear form. That is, we describe trilinear algorithms that compute $\tr(ABC)$ given three matrices (see Fact \ref{fact: trilinear form of mm algorithm}). For completeness, we briefly explain the technique of implicit canceling as well.

Recall that for the multiplication of two $n_0\times n_0$ matrices one needs to compute the sum of $n_0^3$ terms of the form $A_{ij}B_{jk}C_{ki}$, denoted \emph{desirable terms}. All other terms are denoted \emph{undesirable terms}. Trilinear aggregation is a technique to compute several desirable terms using a single multiplication. For example, the product $(A_{0,0}+A_{1,1})(B_{0,1}+B_{1,2})(C_{1,0}+C_{2,1})$ includes two desirable terms $A_{0,0}B_{0,1}C_{1,0}$ and $A_{1,1}B_{1,2}C_{2,1}$, and six undesirable terms.

Undesirable terms need to be canceled. They may be canceled either \emph{explicitly} or \emph{implicitly}. Explicit canceling means subtracting the undesirable terms one by one. However, this requires many multiplications. Thus, correction terms are usually \emph{united}. If, for example, we have to correct terms of the form $a_{0,1}b_{0,1}c_{0,1}$ and $a_{0,1}b_{0,1}c_{0,2}$, they may both be computed using a single multiplication: $a_{0,1}b_{0,1}(c_{0,1}+c_{0,2})$. This union is possible since the terms are \emph{kin}.

Implicit canceling is more complicated. It defines matrices $A^\ast, B^\ast, C^\ast$, such that the elements of $A^\ast$ are linear combinations of elements from $A$, and similarly for $B^\ast, C^\ast$. These matrices also satisfy the property $\tr(A^\ast B^\ast C^\ast) = \tr(ABC)$. The matrices $A^\ast, B^\ast, C^\ast$ have desirable properties that allow for a simpler description of matrix multiplication algorithms. An example of such property is that the sum of each row and column of $A^\ast, B^\ast, C^\ast$ is $0$. This allows undesirable terms to disappear implicitly.

Some undesirable terms cannot be united with other terms and cannot be implicitly canceled. Thus, Pan \cite{pan78, pan80, pan82, laderman1992practical} defines special aggregation tables. When using these tables, many undesirable terms can be united and canceled either explicitly or implicitly. The remaining undesirable terms appear twice in different tables: once with a positive sign and another time with a negative sign. This special structure of the aggregation tables requires the input matrices to be of even dimension.

We use the construction of Hadas and Schwartz \cite{hs23}, which is similar to Pan's \cite{pan82} and obtains the same exponent. The algorithm can be roughly divided into three steps:

\begin{enumerate}
    \itemsep0em
    
    \item Apply a linear transformation $\varphi$ to transform the input matrices $A,B,C\in\mathbb{F}^{n_0 \times n_0}$ into $A^\ast, B^\ast, C^\ast\in\mathbb{F}^{n_0+2\times n_0 + 2}$.
    \item \label{step: sum over aggregation tables} Sum over the aggregation tables. This saves around $\frac{2}{3}$ of the required multiplications, but generates many undesirable terms. Many of them are canceled due to the properties of $A^\ast, B^\ast, C^\ast$ or due to the special structure of the aggregation tables. The remaining undesirable terms can be represented by a tensor of a low rank.
    \item \label{step: cancel remaining terms} Explicitly cancel the remaining terms, utilizing a $\left<2,2,2;7\right>$-algorithm to save multiplications.
\end{enumerate}

The algorithms in \cite{pan82} are $\left<n_0,n_0,n_0,t^{Pan}\right>$-algorithms where $t^{Pan}=\frac{n_0^3}{3} + \frac{15}{4}n_0^2 + \frac{32}{3}n_0 + 9$.

Pan's algorithms split each of the matrices $A^\ast, B^\ast, C^\ast$ into $2\times 2$ blocks of size $\left(\frac{n_0}{2}+1\right)\times \left(\frac{n_0}{2}+1\right)$ each. Thus, we introduce a notation that allows intuitive description of these blocks.

\begin{definition} [cf. \cite{pan78, pan80, pan82, laderman1992practical}]
    Let $i\in[n_0+2]$. We denote $\bar{i}=i+\frac{n_0}{2}+1\mod{n_0+2}$. The value of $n_0$ will always be clear from context.
\end{definition}

The terms added in Step \ref{step: cancel remaining terms} can be written as
\begin{equation*}
\sum_{i,j\in\left[\frac{n_0}{2}+1\right]} \tr\left(d
\begingroup 
\setlength\arraycolsep{2pt}
\begin{pmatrix}
    \gamma_{i,j} A^\ast_{i,j} & A^\ast_{\bar{i}, j} \\ A^\ast_{i, \bar{j}} & A^\ast_{\bar{i},\bar{j}}
\end{pmatrix}
\begin{pmatrix}
    B^\ast_{i,j} & \frac{1}{\gamma_{i,j}}B^\ast_{\bar{i}, j} \\ B^\ast_{i, \bar{j}} & B^\ast_{\bar{i},\bar{j}}
\end{pmatrix}
\begin{pmatrix}
    C^\ast_{i,j} & -C^\ast_{\bar{i}, j} \\ -C^\ast_{i, \bar{j}} & \gamma_{i,j} C^\ast_{\bar{i},\bar{j}}
\end{pmatrix}
\endgroup
\right)
\hfilneg
\end{equation*}
where $d=\left(\frac{n_0}{2}+1\right)$ and $\gamma_{i,j}= 1 - \frac{9}{\frac{n_0}{2}+1}\delta_{i,j}$.

This sum can be computed using $7\cdot\left(\frac{n_0}{2}+1\right)^2$ multiplications by using a $\left<2,2,2;7\right>$-algorithm. Note that even though all $\left<2,2,2;7\right>$-algorithms are de Groote-equivalent \cite{DEGROOTE1978pt1}, there are many different triplets $\left<U,V,W\right>$ that represent a $\left<2,2,2;7\right>$-algorithm. Any of these triplets can be used in Step \ref{step: cancel remaining terms}. Further, not all traces must be computed using the same $\left<2,2,2;7\right>$-algorithm.

We show that we can make some of the $7\cdot\left(\frac{n_0}{2}+1\right)^2$ multiplications kin to multiplications in Step \ref{step: sum over aggregation tables} by carefully choosing the $\left<2,2,2;7\right>$-algorithm we use. Thus, we are able to reduce the total number of multiplications and reduce the exponent of the algorithms.

\section{A New Trilinear Aggregation Based Family of Algorithms}
\label{sec: improving pans algorithms}
In this section we present our trilinear aggregation based algorithms. They have reduced number of multiplications, thus smaller exponents, compared to Pan's \cite{pan82} algorithms (see Table \ref{tab: comparing exponents of small algorithm}). That is, we reduce the number of multiplications required for the same base case from $t^{Pan}=\frac{n_0^3}{3} + \frac{15}{4}n_0^2 + \frac{32}{3}n_0 + 9$ to $t^{New}=\frac{n_0^3}{3} + \frac{15}{4}n_0^2 + \frac{61}{6}n_0 + 8$, where $n_0$ is the matrix dimension ($n_0\neq 16$ is even\footnote{The algorithms in \cite{hs23} require division by $1-\frac{9}{\frac{n_0}{2}+1}$. When $n_0=16$, this is equal to zero and thus the algorithm is not well defined.}).

Our algorithms are based on Pan's algorithms \cite{pan82}, and we describe them using the notations of Hadas and Schwartz \cite{hs23}. We use the same implicit canceling transformations, and the same aggregation tables of \cite{hs23}, which are equivalent to Pan's aggregation tables \cite{pan82} but allow for a more concise description of the algorithm. Our canceling step is similar to that of \cite{hs23}. However, wereas Pan utilizes a single $\left<2,2,2;7\right>$-algorithm, in some cases we choose a different $\left<2,2,2;7\right>$-algorithm for trace computations. This allows creating terms that are kin (identical lines in $U$ and $V$) to terms from the aggregation step. Uniting these terms reduces the required number of multiplications.

We next provide a high level description of the construction. For completeness, Appendix \ref{app: explicit description of algorithms} contains an explicit description of our algorithms in the trilinear form. The $\left<U,V,W\right>$ encoding and decoding matrices are provided as supplemental material to this work \footnote{\supplementalurl}.

Let $n_0\neq 16$ be an even positive integer. We begin by applying the same transformations as Pan \cite{pan82}, and using the aggregation tables of Hadas and Schwartz \cite{hs23} (which are equivalent to Pan's). We describe in detail the cancellation step, which is different.

During the \emph{aggregation step}, some (but not all) multiplications take the form
$$(-A^\ast_{i,j}  + A^\ast_{\j, k} + A^\ast_{k, \i})(B^\ast_{j, \k} + B^\ast_{k,i} + B^\ast_{\i, j}) (-C^\ast_{\k, i} + C^\ast_{i, \j} + C^\ast_{j,k})$$
for $i,j,k\in\left[\frac{n_0}{2} + 1\right]$.
Consider these terms where $i=j=k$ for all $i\in\left[\frac{n_0}{2}+1\right]$, namely
\begin{align}
(-A^\ast_{i,i}  + A^\ast_{\i, i} + A^\ast_{i, \i})(B^\ast_{i, \i} + B^\ast_{i,i} + B^\ast_{\i, i}) (-C^\ast_{\i, i} + C^\ast_{i, \i} + C^\ast_{i,i}) \label{eq: aggregation where i=j=k}
\end{align}

We next show how terms from the \emph{cancellation step} can take a form that is kin to the above terms, hence we can apply Lemma \ref{lemma: kin terms reduce multiplications} to reduce the number of multiplications.

Recall that the cancellation step computes
\begin{align*}
    \tr\left(d
    \begin{pmatrix}
        \gamma_{i,j} A^\ast_{i,j} & A^\ast_{\bar{i}, j} \\ A^\ast_{i, \bar{j}} & A^\ast_{\bar{i},\bar{j}}
    \end{pmatrix}
    \begin{pmatrix}
        B^\ast_{i,j} & \frac{1}{\gamma_{i,j}}B^\ast_{\bar{i}, j} \\ B^\ast_{i, \bar{j}} & B^\ast_{\bar{i},\bar{j}}
    \end{pmatrix}
    \begin{pmatrix}
        C^\ast_{i,j} & -C^\ast_{\bar{i}, j} \\ -C^\ast_{i, \bar{j}} & \gamma_{i,j} C^\ast_{\bar{i},\bar{j}}
    \end{pmatrix}
    \right)
\end{align*}
for all $i,j\in\left[\frac{n_0}{2}+1\right]$, where $d=\frac{n_0}{2}+1$ and $\gamma_{i,j}=1-\frac{9}{d}\delta_{i,j}$.

Notice that $n_0\neq 16$ implies $\gamma_{i,j}\neq 0$ for all $i,j\in\left[\frac{n_0}{2}+1\right]$ and thus this is well defined. 

When $i=j$, this is
\begin{align}
    \tr\left(d
    \begin{pmatrix}
        \gamma_{i,i} A^\ast_{i,i} & A^\ast_{\bar{i}, i} \\ A^\ast_{i, \bar{i}} & A^\ast_{\bar{i},\bar{i}}
    \end{pmatrix}
    \begin{pmatrix}
        B^\ast_{i,i} & \frac{1}{\gamma_{i,i}}B^\ast_{\bar{i}, i} \\ B^\ast_{i, \bar{i}} & B^\ast_{\bar{i},\bar{i}}
    \end{pmatrix}
    \begin{pmatrix}
        C^\ast_{i,i} & -C^\ast_{\bar{i}, i} \\ -C^\ast_{i, \bar{i}} & \gamma_{i,i} C^\ast_{\bar{i},\bar{i}}
    \end{pmatrix}
    \right)
    \label{eq: trace where i=j}
\end{align}

That is, the cancellation step contains a sub-tensor equivalent to the tensor of matrix multiplication. Recall that Pan \cite{pan82} suggested that this trace can be computed using $7$ multiplications (instead of the trivial $8$ multiplication) by utilizing a $\left<2,2,2;7\right>$-algorithm.

We next show that there exist $\left<2,2,2;7\right>$-algorithms such that the first multiplication when computing the trace in Equation (\ref{eq: trace where i=j}) is kin to the multiplication in Equation (\ref{eq: aggregation where i=j=k}). To prove the existence of such algorithms, we use Theorem \ref{theorem: we have 2227 alg with U0=KU and V0=KV}.

\begin{claim}
\label{claim: first row of multiplication}
Let $\left<U,V,W\right>$ be a $\left<2,2,2;7\right>$-algorithm. If $\left<U, V, W\right>$ is used to compute Equation (\ref{eq: trace where i=j}) for some $i\in\left[\frac{n_0}{2}+1\right]$ in the cancellation step, Then the multiplication
\begin{equation}
\begin{aligned}
\label{eq: multiplication in cancellation step}
[U]_0\begin{pmatrix}
    \gamma_{i,i}A^\ast_{i,i} \\ A^\ast_{\bar{i},i} \\ A^\ast_{i,\bar{i}} \\ A^\ast_{\bar{i},\bar{i}}
\end{pmatrix} \cdot [V]_0 \begin{pmatrix}
    B^\ast_{i,i} \\ \frac{1}{\gamma_{i,i}}B^\ast_{\bar{i},i} \\ B^\ast_{i,\bar{i}} \\ B^\ast_{\bar{i},\bar{i}}
\end{pmatrix} \cdot [W]_0 \begin{pmatrix}
    dC^\ast_{i,i} \\ -dC^\ast_{\bar{i},i} \\ -dC^\ast_{i,\bar{i}} \\ d\gamma_{i,i}C^\ast_{\bar{i},\bar{i}}
\end{pmatrix}
\end{aligned}
\end{equation}
appears in the cancellation step.
\end{claim}

\begin{proof}
    This is exactly the first multiplication induced by $\left<U,V,W\right>$ when the input matrices are $$
    \begin{pmatrix}
        \gamma_{i,i} A^\ast_{i,i} & A^\ast_{\bar{i}, i} \\ A^\ast_{i, \bar{i}} & A^\ast_{\bar{i},\bar{i}}
    \end{pmatrix}
    \begin{pmatrix}
        B^\ast_{i,i} & \frac{1}{\gamma_{i,i}}B^\ast_{\bar{i}, i} \\ B^\ast_{i, \bar{i}} & B^\ast_{\bar{i},\bar{i}}
    \end{pmatrix}
    \begin{pmatrix}
        dC^\ast_{i,i} & -dC^\ast_{\bar{i}, i} \\ -dC^\ast_{i, \bar{i}} & d\gamma_{i,i} C^\ast_{\bar{i},\bar{i}}
    \end{pmatrix}
    $$
\end{proof}

\begin{corollary} 
    \label{cor: using 2-2-2-7 creates kin rows}
    If $\left<U, V, W\right>$ is used to compute Equation (\ref{eq: trace where i=j}) for $i\in\left[\frac{n_0}{2}+1\right]$ in the cancellation step, and
    \begin{equation}
    \begin{aligned}
    \label{eq: U0, V0 requirements}
        [U]_0 &= \begin{pmatrix}-\frac{1}{\gamma_{i,i}} & 1 & 1 & 0\end{pmatrix} \\
        [V]_0 &= \begin{pmatrix}1 & \gamma_{i,i} & 1 & 0\end{pmatrix}
    \end{aligned}
    \end{equation}
    then there exist a row in the cancellation step that is kin to a row in the aggregation step.
\end{corollary}

\begin{proof}
The multiplication in Equation (\ref{eq: aggregation where i=j=k}) appears in the aggregation step. 
By Claim \ref{claim: first row of multiplication}, the multiplication in Equation (\ref{eq: multiplication in cancellation step}) appears in the cancellation step. By assumption, Equation (\ref{eq: U0, V0 requirements}) holds. Thus, these two multiplications are kin.
\end{proof}

\begin{claim}
    \label{claim: there exist the 2-2-2-7 alg that creates kin row}
    There exists a $\left<2,2,2;7\right>$-algorithm $\left<U, V, W\right>$ such that Equation (\ref{eq: U0, V0 requirements}) holds.
\end{claim}
\begin{proof}
Let $K_U=\begin{pmatrix}-\frac{1}{\gamma_{i,i}} & 1 \\ 1 & 0\end{pmatrix}$ and $K_V=\begin{pmatrix}1 & \gamma_{i,i} \\ 1 & 0\end{pmatrix}$. Notice that $K_U,K_V$ are invertible. By Theorem \ref{theorem: we have 2227 alg with U0=KU and V0=KV}, there exists a $\left<2,2,2;7\right>$-algorithm such that Equation (\ref{eq: U0, V0 requirements}) is satisfied.
\end{proof}

 Thus, we create algorithms with kin terms. We unite these terms using Lemma \ref{lemma: kin terms reduce multiplications} to reduce the number of multiplications. 

\begin{theorem}
    \label{thm: num of mults of ALG_1}
    There exists an $\left<n_0, n_0, n_0; t\right>$-algorithm such $t=t^{Pan}-\frac{n_0}{2}-1=\frac{n_0^3}{3} + \frac{15}{4}n_0^2 + \frac{61}{6}n_0 + 8$ for every even $n_0 \neq 16$.
\end{theorem}

\begin{definition}
    The algorithms of the previous theorem will be referred as $\szone$.
\end{definition}

\begin{proof} [Proof of Theorem \ref{thm: num of mults of ALG_1}]
    Let $\left<U, V, W\right>$ be the $\left<n_0, n_0, n_0; t^{Pan}\right>$ described in \cite{hs23}, where the traces in Equation (\ref{eq: trace where i=j}) are computed using the $\left<2,2,2;7\right>$-algorithm of Claim \ref{claim: there exist the 2-2-2-7 alg that creates kin row}. By Corollary \ref{cor: using 2-2-2-7 creates kin rows}, the algorithm has $\frac{n_0}{2}+1$ disjoint pairs of kin rows. Thus, by applying Lemma \ref{lemma: kin terms reduce multiplications}, we obtain an $\left<n_0,n_0,n_0;t\right>$-algorithm, where $t=t^{Pan}-\frac{n_0}{2}-1$.
\end{proof}

\begin{corollary}
    \label{cor: optimal exponent of ALG_1}
    The optimal exponent for $\szone$ is $\omega_0 = \log_{44}36110\approx 2.773203$ and is given at $n_0 = 44$. 
\end{corollary}
For completeness, a proof of this claim is provided in Appendix \ref{app: finding optimal n0}.

This concludes the analysis of $\szone$. For the sake of constructing $\sztwo$ (see Section \ref{sec: using kaporins algorithm}), we account for the number of $\left<2,2,2;7\right>$-algorithms within $\szone$.

\begin{claim}
    \label{claim: the number of 2,2,2;7 algorithms in ALG1}

    Let $\left<U,V,W\right>$ be the encoding and decoding matrices of $\szone$ with base case $n_0$. Then $\left<U,V,W\right>$ contains $\frac{n_0^2}{4}+\frac{n_0}{2}$ computations of $\left<2,2,2;7\right>$-algorithms.
\end{claim}

\begin{proof}
    $\szone$ computes $\left(\frac{n_0}{2}+1\right)^2$ traces using $\left<2,2,2;7\right>$-algorithms. The first multiplication of $\frac{n_0}{2}+1$ of these algorithms is united with terms from the aggregation step. Thus, they are no longer $\left<2,2,2;7\right>$-algorithms. Therefore, $\left(\frac{n_0}{2}+1\right)^2-\left(\frac{n_0}{2}+1\right)=\frac{n_0^2}{4}+\frac{n_0}{2}$ of the $\left<2,2,2;7\right>$-algorithms remain intact.
\end{proof}

\section{A Trilinear Aggregation Based Family of Algorithms With Better Complexity but Larger Base Case}
\label{sec: using kaporins algorithm}

    In this section, we describe a second family of matrix multiplication algorithms with lower exponents than those of $\szone$, albeit larger base cases (see Table \ref{tab: comparing exponents of all algorithms}). To this end, we consider two recursive calls of the algorithms of $\szone$. We notice that these algorithms contain sub-tensors that are equivalent to $\left<4,4,4;49\right>$-algorithms. We replace these with a $\left<4,4,4;48\right>$-algorithm \cite{kaporin4_4_4_48, alphaevolve, dumas2025noncommutative}. This process is described formally in the following definition:

\begin{definition}
    Let $\left<U, V, W\right>$ be an $\left<m_0, m_0, m_0;t'\right>$-algorithm from the family $\szone$. 
    
    Denote by $\left<U\otimes U, V\otimes V, W\otimes W\right>$ the $\left<m_0^2, m_0^2, m_0^2, t'^2\right>$ algorithm obtained from a composition of $\left<U, V, W\right>$ with itself. Consider the algorithm $\left<U\otimes U, V\otimes V, W\otimes W\right>$ and replace each occurrence of a $\left<4,4,4;49\right>$-algorithm with a $\left<4,4,4;48\right>$-algorithm \cite{kaporin4_4_4_48, alphaevolve, dumas2025noncommutative}. Denote the resulting algorithm by $\sztwo$.
\end{definition}

\begin{theorem} 
    Let $m_0\neq 16$ be an even number and let $n_0=m_0^2$. $\sztwo$ with base case $n_0$ is an $\left<n_0, n_0, n_0; t\right>$-algorithm where $t=\frac{n_0^3}{9} + \frac{5n_0^{2.5}}{2}+\frac{187n_0^2}{9} + \frac{244n_0^3}{1.5} + \frac{1468n_0}{9}+\frac{488\sqrt{n_0}}{3} + 64$ 
\end{theorem}

\begin{proof}
    Let $\left<U, V, W\right>$ be an $\left<n_0, n_0, n_0; t'\right>$-algorithm from the family $\szone$.
    By Claim \ref{claim: the number of 2,2,2;7 algorithms in ALG1}, one call to $\left<U,V,W\right>$ requires $h$ computations of $\left<2,2,2;7\right>$-algorithms, where $h=\frac{m_0^2}{4}+\frac{m_0}{2}$. It follows that one recursive step of $\left<U\otimes U, V\otimes V, W\otimes W\right>$ requires $h^2$ computations of $\left<4,4,4;49\right>$-algorithms. Thus, $\sztwo$ requires $t=t'^2-h^2=\frac{m_0^6}{9} + \frac{5m_0^5}{2}+\frac{187m_0^4}{9} + \frac{244m_0^3}{3} + \frac{1468m_0^2}{9}+\frac{488m_0^3}{3} + 64$ multiplications.
\end{proof}

\begin{corollary}
    \label{cor: optimal exponent of ALG_2}
    The optimal exponent given by $\sztwo$ is $\omega_0 = \log_{1936} 1303676064 \approx 2.773177$ and is obtained for $n_0=44^2=1936$.  
    
    For completeness, a proof of this claim is provided in Appendix \ref{app: finding optimal n0}.\looseness=-1
\end{corollary}

\section{Discussion}
\label{sec: discussion}
We describe two new families of algorithms based on Pan's 1982 algorithms \cite{pan82}. The first (recall Table \ref{tab: comparing exponents of small algorithm}) has the smallest exponents for many even feasible base cases of size $n_0 \geq 28$. One of the algorithms in this family is a $\left<44,44,44; 36110\right>$-algorithm that has complexity of $O\left(n^{2.773203}\right)$, the asymptotically fastest matrix multiplication algorithm with base case size smaller than $1000$. However, its additive complexity is quite large. That is, the leading coefficient (namely, the hidden constant in the $O$-notation) of our algorithm is about $736$, making it impractical. Similarly, Pan's algorithm \cite{pan82} has a leading coefficient of $734$ \cite{hs23}. Hadas and Schwartz \cite{hs23} reduced the leading coefficient of Pan's 1982 algorithms \cite{pan82} to about $8$. We reduce the leading coefficient to about $8$ by using a similar technique (see Table \ref{tab:leading coefficients of algorithms}). We detail the construction and analysis in Appendix \ref{app: sparse decomposition}.\looseness=-1

Our second family of algorithms (recall Table \ref{tab: comparing exponents of all algorithms}) has better exponents, with larger base cases. The best algorithm of that family has base case $1936$ and complexity $O(n^{2.773177})$. This is the best exponent of all known feasible algorithms.

\paragraph{Future Work}
Our $\left<44,44,44;36110\right>$-algorithm has an exponent of $\omega_0\approx 2.773203$. The lower bound of Bl{\"a}ser is $4708$ multiplications \cite{blaser2003_lower_bound}. Similar gaps exists for other dimensions. Further research is needed to close this gap by either finding new algorithms or proving larger bounds.

We improve the additive complexity of our algorithms, and significantly reduce their leading coefficient. This was achieved by using the sparse decomposition method of \cite{bs19, hs23}. However, this method increases IO-complexity. That is, it increases communication costs within memory hierarchy and between parallel processors \cite{bs19}. The leading coefficient of the algorithms can also be reduced using the methods of \cite{ks} and \cite{martensson_number_of_the_beast}, resulting in an improved leading coefficient while preserving communication costs. We leave this for future work.

The next step is to implement and benchmark our algorithms against state-of-the-art implementations of other algorithms. We also leave for future work the parallelization of our algorithms, their adaptation to matrices of arbitrary dimensions, and the study of their numerical stability, which can be handled in a way similar to the works of \cite{noa_pebbling, dumas2024strassen}.

Drevet et al. \cite{drevet2011optimization} use the trilinear aggregation method to construct algorithms for even base cases, and then adapt them to odd base cases. It may be possible to use their method to adapt our algorithms for odd base cases. Generalizing our algorithms for rectangular matrix multiplication is of interest for practical use as well. We leave this for future research.

To obtain base cases with minimal number of multiplications, one can combine several algorithms (cf. Drevet et al. \cite{drevet2011optimization}). For example, to multiply $88\times 88$ matrices with low multiplication count, our method yields an $\left<88, 88, 88; 257100\right>$-algorithm while a composition of Strassen's algorithm with our $\left<44, 44, 44, 36110\right>$-algorithm requires $252770$ multiplications (see Appendix \ref{app: composition of algorithms}). When the base case is not an exact multiple of other algorithms' base cases, techniques such as zero padding \cite{drevet2011optimization} may be used to obtain new improved base cases.
We leave for future research the construction of further algorithms with low exponents by combining our algorithms with existing ones.

Recursive fast matrix multiplication algorithms often switch to the na\"{i}ve algorithm on small blocks, as for smaller blocks, the na\"{i}ve algorithm outperforms the recursive algorithm. Similarly, one can call different fast matrix multiplication algorithms at distinct recursive levels to optimize the complexity \cite{noa_pebbling, moran_2_2_sub_blocks}. This allows reducing the additive complexity. We leave this analysis for future research.

\begin{table}
    \centering
    \begin{tabular}{|c|c|c|c|c|c|c|}
    \hline
    \multirow{3}{*}{Base Case $n_0$} & \multicolumn{3}{c|}{\cite{pan82}} & \multicolumn{3}{c|}{[Here]} \\
    \cline{2-7}
    & \multirow{2}{*}{Exponent $\omega_0$} & \multicolumn{2}{c|}{Leading Coefficient} & \multirow{2}{*}{Exponent $\omega_0$} & \multicolumn{2}{c|}{Leading Coefficient} \\
    \cline{3-4}\cline{6-7}
    && Original & Reduced \cite{hs23} && Original & Reduced \\
    \hline\hline

$20$ & $2.799602$ & $308.098$ & $8.175$ & $2.798764$ & $308.588$ & $8.419$ \\ 
\hline
$30$ & $2.778337$ & $486.410$ & $8.120$ & $2.777967$ & $486.035$ & $8.265$ \\ 
\hline
$40$ & $2.773655$ & $665.605$ & $8.091$ & $2.773450$ & $664.699$ & $8.193$ \\ 
\hline
$42$ & $2.773444$ & $701.493$ & $8.087$ & $2.773258$ & $700.505$ & $8.183$ \\ 
\hline
$44$ & $2.773372$ & $737.392$ & $8.083$ & $2.773203$ & $736.328$ & $8.174$ \\ 
\hline
$46$ & $2.773412$ & $773.301$ & $8.080$ & $2.773258$ & $772.166$ & $8.165$ \\ 
\hline
$48$ & $2.773543$ & $809.217$ & $8.076$ & $2.773403$ & $808.017$ & $8.158$ \\ 
\hline
$50$ & $2.773749$ & $845.141$ & $8.073$ & $2.773620$ & $843.880$ & $8.151$ \\ 
\hline
$60$ & $2.775496$ & $1024.842$ & $8.062$ & $2.775408$ & $1023.328$ & $8.124$ \\ 
\hline
    
    \end{tabular}
    \caption{Comparison of Pan's algorithms \cite{pan82} and our algorithms (Section \ref{sec: improving pans algorithms}). We use the decomposition techniques of Hadas and Schwartz \cite{hs23} to reduce the leading coefficient. Our algorithms obtain leading coefficients that are comparable to \cite{hs23} while having better exponent.}
    
    \label{tab:leading coefficients of algorithms}
\end{table}

{
\printbibliography
}
\newpage
\appendix
\section{$\left<2,2,2;7\right>$-algorithms}
\label{app: 222-7 algorithms}

We next prove Theorem \ref{theorem: we have 2227 alg with U0=KU and V0=KV}, which states the given two invertible matrices $K_U,K_V\in\mathbb{R}^{2\times 2}$, there exists a $\left<2,2,2;7\right>$-algorithm with encoding/decoding matrices $\left<U,V,W\right>$ such that $[U]_0=\vec{K_U}$ and $[V]_0=\vec{K_V}$.

We begin by introduction the inverse of the vectorization operator (Definition \ref{def: vectorization}).

\begin{definition}
    Let $\ell \in \mathbb{N}$ and $v \in \mathbb{F}^{m_0^\ell \cdot  n_0^\ell}$. The \emph{matrix represented by $v$} is $\mathcal{M}_{m_0,n_0}(v)\in \mathbb{F}^{m_0^\ell \times n_0 ^{\ell}}$ satisfying $\overrightarrow{\mathcal{M}_{m_0,n_0}(v)}=v$.
    
    For ease of notation, we omit the parameters $n_0,m_0$ when they are clear from context. To simplify notation further, if $v \in \mathbb{F}^{1\times m_0^\ell\cdot n_0^\ell}$, we denote $\mathcal{M}_{m_0,n_0}(v)=\mathcal{M}_{m_0,n_0}(v^T)$.
\end{definition}

We next introduce the following lemma that allows modifying the first row of $U$ without affecting $W$.
\begin{lemma}
    \label{lemma: modify first row of u only}
    Let $\left<U,V,W\right>$ be a $\left<2,2,2;7\right>$-algorithm such that $\mathcal{M}([U]_0)$ is invertible and let $K\in\mathbb{F}^{2\times 2}$ be an invertible matrix. Then there exists a $\left<2,2,2;7\right>$-algorithm $\left<U',V',W'\right>$ such that $W'=W$ and $[U']_0=\vec{K}$.
    Further, if $\mathcal{M}([V]_0)$ is invertible then $\mathcal{M}([V']_0)$ is also invertible.
\end{lemma}

\begin{claim}
    [\cite{strassen1969}]
    \label{claim: Strassens Algorithm} If
    $$U^{Str}=\begin{pmatrix}
        1 & 0 & 0 & 1\\
        0 & 0 & 0 & 1\\
        0 & 1 & 0 & 1\\
        1 & -1 & 0 & 0\\
        1 & 0 & 1 & 0\\
        1 & 0 & 0 & 0\\
        0 & 0 & 1 & -1
    \end{pmatrix},
    V^{Str}=\begin{pmatrix}
        1 & 0 & 0 & 1\\
        1 & 0 & 1 & 0\\
        0 & 0 & -1 & 1\\
        0 & 0 & 0 & 1\\
        1 & -1 & 0 & 0\\
        0 & -1 & 0 & -1\\
        1 & 0 & 0 & 0
    \end{pmatrix},
    W^{Str}=\begin{pmatrix}
    1 & 0 & 0 & 1\\
    -1 & 1 & 0 & 0\\
    -1 & 0 & 0 & 0\\
    -1 & 0 & -1 & 0\\
    0 & 0 & 0 & -1\\
    0 & 0 & -1 & 1\\
    0 & 1 & 0 & 1
    \end{pmatrix}$$
    then $\left<U^{Str}, V^{Str}, W^{Str}\right>$ is a $\left<2,2,2;7\right>$-algorithm. Further, both $\mathcal{M}([U^{Str}]_0)$ and $ \mathcal{M}([V^{Str}]_0)$ are invertible.
\end{claim}

\begin{proof}[Proof of Theorem \ref{theorem: we have 2227 alg with U0=KU and V0=KV}]
    Let $\left<U^{Str},V^{Str},W^{Str}\right>$ be the $\left<2,2,2;7\right>$-algorithm described in Claim \ref{claim: Strassens Algorithm}. By Lemma \ref{lemma: modify first row of u only}, there exists a triplet $\left<U', V', W'\right>$ representing a $\left<2,2,2;7\right>$-algorithm such that $\mathcal{M}([U'']_0)=K_U$ and $\mathcal{M}([V']_0)$ is invertible. By Corollary \ref{cor: cyclicness of UVW},  $\left<V', U', W'\right>$ is a $\left<2,2,2;7\right>$-algorithm.
    
    Applying Lemma \ref{lemma: modify first row of u only} again we obtain $\left<V, W, U\right>$ representing a $\left<2,2,2;7\right>$-algorithm, with $\mathcal{M}([V]_0)=K_V$ and $\mathcal{M}([U]_0)=\mathcal{M}([U']_0)=K_U$. Finally, by Corollary \ref{cor: cyclicness of UVW}, $\left<U,V,W\right>$ is a $\left<2,2,2;7\right>$-algorithm.
\end{proof}

To prove the lemma, we create $\left<2,2,2;7\right>$-algorithms using the de Groote operator.
Before formally introducing the de Groote operator, we provide a claim regarding the relations between tensor operations and vectorized matrices.

\begin{claim}
    Let $A,K,L\in\mathbb{F}^{n_0 \times n_0}$. Then $(K\otimes L^T)\vec{A}=\overrightarrow{KAL}$.
\end{claim}
\begin{proof}
    Since $K\otimes L^T$ is a linear transformation, it is sufficient to show the correctness of the claim for $A=e_ie_j^T$ for all $i,j\in[n_0]$. Let $\vec{A} = e_i \otimes e_j$. Then $KAL=k_il_j^T$ where $k_i$ is the $i$-th column of $K$ and $l_j$ is the $j$-th column of $L^T$. This implies $\overrightarrow{KAL}=k_i\otimes l_j$.
    
    In addition, $(K\otimes L^T)\vec{A}$ is the $n_0\cdot i + j$-th column of $K\otimes L^T$. By the definition of $K\otimes L^T$, this is also $k_i \otimes l_j$.
\end{proof}
\begin{corollary}
    \label{cor: multiply from the left or right by K}
    Let $A,K\in\mathbb{F}^{n_0\times n_0}$. Then
    \begin{align*}
        \overrightarrow{AL} &= \left(I_{n_0} \otimes L^T \right)\vec{A} \\
        \overrightarrow{KA} &= \left(K\otimes I_{n_0}\right)\vec{A}
    \end{align*}
\end{corollary}

We now introduce the de Groote operator.

\begin{claim}[\cite{DEGROOTE1978pt1}] \label{claim:version of algorithm}
    Let $n\in\mathbb{N}$, Let $\left<U, V, W\right>$ be a $\left<n_0,n_0,n_0;t\right>$-algorithm. Let $K\in\mathbb{F}^{n_0\times n_0}$ be an invertible matrix.
    Let $U'=U\cdot(I_{n_0} \otimes K^T),V'=V\cdot(K^{-1}\otimes I_n)$ and $W'=W$.
    
    Then $\left<U', V', W'\right>$ is also a $\left<n_0,n_0,n_0;t\right>$-algorithm.
\end{claim}

\begin{proof}
    Let $A,B\in\mathbb{F}^{n_0\times n_0}$. Then
    \begin{align*}  
        W^T\cdot(U'A \odot V'B) & = W^T \left(U (I_{n_0} \otimes K^T) \vec{A} \odot V(K^{-1} \otimes I_{n_0})\vec{B}\right) \\
        & = W^T \left( U\left(\overrightarrow{AK}\right)\odot V\left(\overrightarrow{K^{-1}B}\right)\right) \\
        & = \overrightarrow{(AKK^{-1}B)^T} = \overrightarrow{(AB)^T}
    \end{align*}
    where the second equality follows from Corollary \ref{cor: multiply from the left or right by K}.
\end{proof}

\begin{proof}[Proof of Lemma \ref{lemma: modify first row of u only}]
    Denote $T_U=\mathcal{M}([U]_0)$. Let $R=K^T\cdot T_U^{-T}$. By assumption, $T_U, K$ are invertible and therefore $R$ is invertible. 
    We define $$\left<U',V',W'\right>=\left<U \cdot (I_2 \otimes R^T), V\cdot (R^{-1}\otimes I_2), W\right>$$

    By Claim \ref{claim:version of algorithm}, $\left<U', V' ,W'\right>$ is a $\left<2,2,2;7\right>$-algorithm.

    We now show that $[U']_0=\vec{K}^T$.
    Notice that $[U']_0^T = (I_2\otimes R)[U]_0^T$. By Corollary \ref{cor: multiply from the left or right by K}, $\mathcal{M}([U']_0)=\mathcal{M}([U]_0)R^T$. Thus, $\mathcal{M}([U']_0)=K$.
 
    We now assume that $\mathcal{M}([V_0])$ is invertible and show that $\mathcal{M}([V_0'])$ is invertible.
    Recall that $[V']_0^T =(R^{-T} \otimes I_2) [V]_0^T$. By Corollary \ref{cor: multiply from the left or right by K}, $\mathcal{M}([V']_0) = R^{-T} \mathcal{M}([V]_0)$.    
    Recall that $R$ is invertible.
    Thus, $\mathcal{M}([V']_0)$ is invertible.

\end{proof}

\section{Explicit Description of Our Algorithms}
\label{app: explicit description of algorithms}
For completeness, we provide a description of the trilinear aggregation based algorithms described in Section \ref{sec: improving pans algorithms}. The description is provided in the trilinear form (see Fact \ref{fact: trilinear form of mm algorithm}). For a more detailed description, please see \cite{pan82, hs23}. We use the version and notations of \cite{hs23}. The algorithm consists of three stages: transformation, aggregation and cancellation.

Begin by transforming the three matrices $A,B,C$. As Pan notes \cite{pan82}, the transformation is not unique. We explicitly describe one possible transformation using the notations of \cite{kaporin_implementation}.

Let $u\in\mathbb{F}^{\frac{n_0}{2}+1}$ be the all $1$s vector. Let
\begin{align*}
L&=\left(\begin{array}{c}
     I_{\frac{n_0}{2}}\\
     \hline
     -u^T
\end{array}\right) \\
R&=\left(\begin{array}{c|c}
     I_{\frac{n_0}{2}} - \frac{1}{\frac{n_0}{2}+1}uu^T& - \frac{1}{\frac{n_0}{2}+1}u  
\end{array}\right)
\end{align*}

First, we define the transformation $\varphi(X)=(I_2\otimes L)X(I_2 \otimes R)$ and compute
\begin{align*}
    A^\ast &= \varphi(A) \\ 
    B^\ast &= \varphi(B) \\
    C^\ast &= \varphi(C)
\end{align*}

Next, define 
\begin{align*}
    \dot{S} &= \left\{(i,j,k),(\bar{i},\bar{j},\bar{k}):i,j,k\in\left[\frac{n_0}{2}+1\right]\right\} \\
    \dot{S}^1 &= \left\{(i,j,k)\in \left[\frac{n_0}{2}+1\right]^3:i=j=k\right\} \\
    \dot{\hat{S}} &= \left\{(i,j,k), (\bar{i}, \bar{j},\bar{k})\in \left[\frac{n_0}{2}+1\right]^3:i\leq j < k \text{ or } k<j\leq i \right\} \\
    \tilde{S} &= \left\{(i,j)\in \left[\frac{n_0}{2}+1\right]^2:i\neq j\right\} \\
    \tilde{S}^1 &= \left\{(i, i): i\in\left[\frac{n_0}{2}+1\right]\right\} \\
    \gamma &= 1-\frac{9}{\frac{n_0}{2}+1} \\
    d &= \frac{n_0}{2} + 1
\end{align*}

And 
\begin{align*}
    \mathbf{R}(i) &= \left(A^{\ast}_{\bar{i},i} + A^{\ast}_{i,\bar{i}} - A^{\ast}_{i,i}\right) \cdot \left(B^{\ast}_{\bar{i},i} + B^{\ast}_{i,\bar{i}} + B^{\ast}_{i,i}\right)  \\& \qquad\qquad\qquad\qquad\quad\cdot \left(\frac{C^{\ast}_{\bar{i},\bar{i}} d \left(1 - \gamma\right)}{\gamma} - \frac{C^{\ast}_{\bar{i},i} \left(\gamma - d\right)}{\gamma} - \frac{C^{\ast}_{i,\bar{i}} \left(- \gamma + d\right)}{\gamma} + C^{\ast}_{i,i} \left(1 - d\right)\right) \\
    &+\left(A^{\ast}_{i,\bar{i}}\right) \cdot \left(\frac{B^{\ast}_{\bar{i},\bar{i}} \left(- \gamma - 1\right)}{\gamma} - \frac{B^{\ast}_{\bar{i},i}}{\gamma} + B^{\ast}_{i,\bar{i}} \left(1 - \frac{1}{\gamma^{2}}\right) + \frac{B^{\ast}_{i,i} \left(\gamma - 1\right)}{\gamma}\right) \\& \qquad\qquad\qquad\qquad\qquad\qquad\qquad\qquad\qquad\qquad\qquad\cdot \left(C^{\ast}_{\bar{i},\bar{i}} d + C^{\ast}_{\bar{i},i} d + \frac{C^{\ast}_{i,\bar{i}} d}{\gamma} + C^{\ast}_{i,i} d\right) \\
    &+ \left(A^{\ast}_{i,\bar{i}} + A^{\ast}_{i,i} \gamma\right) \cdot \left(\frac{B^{\ast}_{\bar{i},\bar{i}} \left(\gamma + 1\right)}{\gamma} + \frac{B^{\ast}_{\bar{i},i} \left(\gamma + 1\right)}{\gamma} + \frac{B^{\ast}_{i,\bar{i}}}{\gamma^{2}} + \frac{B^{\ast}_{i,i}}{\gamma}\right) \cdot \left(\frac{C^{\ast}_{i,\bar{i}} d}{\gamma} + C^{\ast}_{i,i} d\right) \\
    &+ \left(A^{\ast}_{\bar{i},i} + A^{\ast}_{i,i} \left(- \gamma - 1\right)\right) \cdot \left(B^{\ast}_{\bar{i},\bar{i}} + B^{\ast}_{\bar{i},i} + \frac{B^{\ast}_{i,\bar{i}}}{\gamma} + B^{\ast}_{i,i}\right) \cdot \left(\frac{C^{\ast}_{i,\bar{i}} d}{\gamma^{2}} + C^{\ast}_{i,i} \left(d + \frac{d}{\gamma}\right)\right) \\
    &+ \left(A^{\ast}_{\bar{i},\bar{i}} + A^{\ast}_{\bar{i},i} - \frac{A^{\ast}_{i,\bar{i}}}{\gamma} - A^{\ast}_{i,i}\right) \cdot \left(B^{\ast}_{\bar{i},\bar{i}} \left(- \gamma - 1\right) - \frac{B^{\ast}_{i,\bar{i}}}{\gamma}\right) \cdot \left(\frac{C^{\ast}_{\bar{i},\bar{i}} d \left(\gamma - 1\right)}{\gamma} - \frac{C^{\ast}_{\bar{i},i} d}{\gamma}\right) \\
    &+ \left(A^{\ast}_{\bar{i},i} - A^{\ast}_{i,i}\right) \cdot \left(B^{\ast}_{\bar{i},\bar{i}} \left(- \gamma - 1\right) - B^{\ast}_{\bar{i},i} + \frac{B^{\ast}_{i,\bar{i}} \left(- \gamma - 1\right)}{\gamma} - B^{\ast}_{i,i}\right) \\&\qquad\qquad\qquad\qquad\qquad\qquad\cdot \left(\frac{C^{\ast}_{\bar{i},\bar{i}} d \left(1 - \gamma\right)}{\gamma} + \frac{C^{\ast}_{\bar{i},i} d}{\gamma} - \frac{C^{\ast}_{i,\bar{i}} d \left(\gamma - 1\right)}{\gamma^{2}} + \frac{C^{\ast}_{i,i} d}{\gamma}\right) \\
    &+ \left(A^{\ast}_{\bar{i},\bar{i}} + \frac{A^{\ast}_{i,\bar{i}} \left(- \gamma - 1\right)}{\gamma}\right) \cdot \left(- B^{\ast}_{\bar{i},\bar{i}} + \frac{B^{\ast}_{i,\bar{i}} \left(\gamma - 1\right)}{\gamma}\right) \cdot \left(\frac{C^{\ast}_{\bar{i},\bar{i}} d}{\gamma} - C^{\ast}_{\bar{i},i} \left(- d - \frac{d}{\gamma}\right)\right)
\end{align*}

Finally, compute
\begin{align*}
    \tr(ABC) &= \sum_{(i,j,k)\in \dot{\hat{S}}} (A^\ast_{i,j} + A^\ast_{j,k} + A^\ast _{k,i})(B^\ast_{j,k} + B^\ast _{k,i} + B^\ast _{i,j})(C^\ast _{k,i} + C^\ast_{i,j} + C^\ast_{j,k}) \\
    &+ \sum_{(i,j,k)\in \dot{S} \setminus \dot{S}^1} (-A^\ast_{i,j} + A^\ast_{\bar{j},k} + A^\ast _{k,\bar{i}})(B^\ast_{j,\bar{k}} + B^\ast _{k,i} + B^\ast _{\bar{i},j})(-C^\ast _{\bar{k},i} + C^\ast_{i,\bar{j}} + C^\ast_{j,k}) \\
    & - \left(\frac{n_0}{2}+1\right)\sum_{(i,j)\in \tilde{S} \setminus \tilde{S}^1} \tr\left(
    \begin{pmatrix}
        A^\ast_{i,j} & A^\ast_{\bar{i}, j} \\ A^\ast_{i, \bar{j}} & A^\ast_{\bar{i},\bar{j}}
    \end{pmatrix}
    \begin{pmatrix}
        B^\ast_{i,j} & B^\ast_{\bar{i}, j} \\ B^\ast_{i, \bar{j}} & B^\ast_{\bar{i},\bar{j}}
    \end{pmatrix}
    \begin{pmatrix}
        C^\ast_{i,j} & -C^\ast_{\bar{i}, j} \\ -C^\ast_{i, \bar{j}} & C^\ast_{\bar{i},\bar{j}}
    \end{pmatrix}
    \right) \\
    & + \sum_{i\in\left[\frac{n_0}{2}+1\right]} \mathbf{R}(i)
\end{align*}
where the traces in the third row may be computed using any $\left<2,2,2;7\right>$-algorithm.

\section{Finding the Optimal Base Case Size}
\label{app: finding optimal n0}
In Sections \ref{sec: improving pans algorithms} and \ref{sec: using kaporins algorithm} we introduce two new families of algorithms, and introduce formulas for the asymptotic complexity, as a function of the base case $n_0$. We claim that the first and second families obtain an optimal exponent at $n_0=44$ and $n_0=1936$, respectively. For completeness, we provide formal proofs for these claims.

\begin{proof}[Proof of Corollary \ref{cor: optimal exponent of ALG_1}]
    To find the minimum of $f(n_0)=\log_{n_0}\left(\frac{n_0^3}{3} + \frac{15}{4}n_0^2 + \frac{61}{6}n_0 + 8\right)$ where $n_0\neq 16$ is a positive even number, we first show that the minimum is obtained for some $n_0 < 243$. Then, we search all possible even values for $n_0\neq 16$ where $0<n_0<243$ and conclude that the minimum is $f(44)\approx2.773203$.

    For every even $n_0 > 3^5=243$, $$f(n_0) > \log_{n_0}\frac{n_0^3}{3} = 3 - \log_{n_0}3 \geq 3-\log_{3^5}3=2.8>f(44)$$ Thus, the minimum is obtained for some $n_0 < 243$. Checking all even values of $n_0\neq 16$ where $0<n_0<243$, we see that the minimum is indeed obtained for $n_0=44$.
\end{proof}

\begin{proof}[Proof of Corollary \ref{cor: optimal exponent of ALG_2}]
    
\end{proof}
    To find the minimum of $f(m_0)=\log_{m_0^2}\left(\frac{m_0^6}{9} + \frac{5m_0^5}{2}+\frac{187m_0^4}{9} + \frac{244m_0^3}{3} + \frac{1468m_0^2}{9}+\frac{488m_0^3}{3} + 64\right)$ where $m_0$ is a positive even number, we first show that the minimum is obtained for some $m_0 < 243$. Then, we search all possible even values for $n_0\neq 16$ where $0<m_0<243$ and conclude that the minimum is $f(44)\approx2.773177$.

    For every even $m_0 > 3^5=243$, $$f(m_0) > \log_{m_0^2}\frac{m_0^6}{9} = 3 - \log_{m_0}3 \geq 3-\log_{3^5}3=2.8>f(44)$$ Thus, the minimum is obtained for some $m_0 < 243$. Checking all even values of $m_0\neq 16$ where $0<m_0<243$, we see that the minimum is indeed obtained for $n_0=44$.

\section{Additive Complexity: Reducing the Leading Coefficient}
\label{app: sparse decomposition}
We next describe how to reduce the leading coefficient of the algorithms by improving the additive complexity. We begin by presenting the key ideas behind the sparse decomposition technique \cite{bs19}. This technique is a generalization of the alternative basis method \cite{ks}. 

\subsection{Decomposed Bilinear Algorithms}
\begin{definition}[\cite{bs19}]
A decomposed recursive bilinear algorithm $\left<U_\phi, V_\psi, W_\nu\right>_{\phi,\psi,\nu}$ has the following form:
\begin{enumerate}
    \item Apply fast basis transformations $\phi, \psi$ on inputs $A,B$. Denote the transformed inputs as $\tilde{A},\tilde{B}$.
    \item Apply a recursive bilinear algorithm $\left<U_\phi, V_\psi, W_\nu\right>$ on $\tilde{A},\tilde{B}$ and obtain $\tilde{C}$.
    \item Apply a fast basis transformation $\nu^T$ to $\tilde{C}$ to obtain the result $C$.
\end{enumerate}

For ease of notation, if $\phi=\psi=\nu$ we simply denote $\left<U_\phi, V_\phi, W_\phi\right>_\phi$.
\end{definition}

\begin{definition}[\cite{bs19}]
    Let $\varphi_1:\mathbb{R}^{s_1}\to\mathbb{R}^{s_2}$ be a linear transformation. Let $l\in\mathbb{N}$ and denote $S_1=(s_1)^l,S_2=(s_2)^l$. Let $v\in\mathbb{R}^{s_1}$. Denote $v^{(i)}=\begin{pmatrix}
        v_{\frac{S_1}{s_1}(i)},\dots,v_{\frac{S_1}{s_1}(i+1)-1}
    \end{pmatrix}$. The linear map $\varphi_l(v)$ is recursively defined as
    $$\varphi_l(v)=\begin{pmatrix}
        \varphi_{l-1}(v^{(0)}) \\ \vdots \\ \varphi_{l-1}(v^{(s_1-1)})
    \end{pmatrix}$$
\end{definition}

\begin{claim}[\cite{bs19}]
    Let $\left<U,V,W\right>$ be a matrix multiplication algorithm. Let $\left<U_\phi, V_\psi, W_\nu\right>_{\phi,\psi,\nu}$ be a decomposed recursive bilinear algorithm. If $U=U_\phi\phi, V=V_\psi\psi$, and $W=W_\nu\nu$ then $\left<U_\phi, V_\psi, W_\nu\right>_{\phi,\psi,\nu}$ computes matrix multiplication.
\end{claim}

\subsection{Analysis}
We next analyze the additive complexity of our decomposed recursive bilinear algorithm. To this end, we separately analyze the two components making up the decomposed recursive bilinear algorithm: the fast transformations, and the recursive bilinear part. In our case it is sufficient to analyze in the case where $\phi=\psi=\nu$, as we will use the same transformation for both encoding and decoding. We provide the required definitions and claims, taken from \cite{bs19}.

\begin{definition}[\cite{bs19}]
    Let $U\in\mathbb{R}^{t\times r}$. We define the number of non-zeros ($nnz$), the number of non-singletons ($nns$), the number of rows ($nrows$), and the number of columns ($ncols$) as follows:
    \begin{align*}
        nnz(U) &= \left|\left\{(i,j)\in[t]\times [r] : U_{i,j} \neq 0\right\}\right| \\
        nns(U) &= \left|\left\{(i,j)\in[t]\times [r] : U_{i,j} \notin \left\{-1,0,1\right\}\right\}\right| \\
        nrows(U) &= t \\
        ncols(U) &= r
    \end{align*}
\end{definition}

\begin{claim} [\cite{bs19}]
\label{claim: qu qv qw number of linear operations}
    Let $\left<U,V,W\right>$ be a billinear algorithm. Let $q_U, q_V, q_W$ be the number of linear operations incurred by encoding/decoding using $U,V,W$ respectively. Then
    \begin{align*}
        q_U &= nnz(U)+nns(U)-nrows(U) \\
        q_V &= nnz(V)+nns(V)-nrows(V) \\
        q_W &= nnz(W)+nns(W)-ncols(W)
    \end{align*}
\end{claim}

\begin{claim} [\cite{bs19}]
    Let $U,V,W\in\mathbb{R}^{t_0\times s_0}$. Let $ALG=\left<U,V,W\right>$ be a recursive bilinear algorithm. Let $q=q_U+q_V+q_W$ be the number of linear operations performed at the base case. Then the additive complexity of $ALG$ is
    $$F_{ALG}(s)=\left(1+\frac{q}{t_0-s_0}\right)s^{\log_{s_0}t_0}-\left(\frac{q}{t_0-s_0}\right)s$$
\end{claim}

\begin{claim} [\cite{bs19}]
    Let $\varphi_1:\mathbb{R}^{s_1} \to \mathbb{R}^{s_2}$ be a linear transformation where $s_1\neq s_2$. Let $q_\varphi=nnz(\varphi)+nns(\varphi)-nrows(\varphi)$ be the additive complexity of $\varphi_1$. Then the additive complexity of $\varphi_l(v)$ is
    $$F_\varphi(s_1^l)=\frac{q_\varphi}{s_1-s_2}(s_1^l-s_2^l)$$
\end{claim}

Using the previous claims, we can give an exact formula for the runtime of a recursive bilinear algorithm.

\begin{theorem} [\cite{bs19}]
    Let $ALG=\left<U_\phi, V_\phi, W_\phi\right>_\phi$ be a decomposed recursive matrix multiplication algorithm, where $U_\varphi, V_\varphi, W_\varphi\in\mathbb{R}^{t_0\times s_0}$ and $\varphi\in\mathbb{R}^{s_0 \times n_0^2}$. Denote $\omega_0=\log_{n_0}t_0$. Then 
    \begin{align*}
        F_{ALG}(n) &= \left(\frac{q_{U_\phi}+q_{V_\phi}+q_{W_\phi}}{t_0-s_0}+1\right)n^{\omega_0} \\
        &+ \left(\frac{2q_\phi+q_{\phi^T}}{s_0-n_0^2} - \frac{q_{U_\phi}+q_{V_\phi}+q_{W_\phi}}{t_0-s_0}\right)n^{\log_{n_0}s_0} \\
        &- \frac{2q_\phi+q_{\phi^T}}{s_0-n_0^2}n^2
    \end{align*}
\end{theorem}

\begin{corollary} [\cite{bs19}]
    \label{cor: leading coefficient of bs}
    The leading coefficient is $\frac{q_{U_\phi}+q_{V_\phi}+q_{W_\phi}}{t_0-s_0}+1$.
\end{corollary}

\subsection{Decomposing $\szone$ Algorithms}
We next reduce the additive complexity of our algorithm, making the practical. This is obtained by finding a sparse decomposition \cite{bs19} of our algorithms. Hadas and Schwartz \cite{hs23} reduce the leading coefficient of Pan's 1982 algorithms by finding a sparse decomposition of Pan's algorithms. Their decomposition is based on the transformations embedded in Pan's algorithm. The improved coefficient is obtained by applying the transformations separately from the rest of the algorithm.

We similarly reduce the leading coefficient of our algorithms to about $8$ (recall Table \ref{tab:leading coefficients of algorithms}). To this end, we describe a decomposition of our algorithms by using a technique similar to that of Hadas and Schwartz \cite{hs23}. 

The description of our algorithms (Section \ref{sec: improving pans algorithms}) begins by applying a transformation $\varphi$ to input and output matrices. Thus, by computing the transformation separately to the recursive bilinear part, we obtain a decomposed recursive algorithm of the form $\left<U_\varphi, V_\varphi, W_\varphi\right>_\varphi$. The resulting leading coefficients is similar to that of Hadas and Schwartz and is found in Table \ref{tab:leading coefficients of algorithms}. The leading coefficient is obtained from Table \ref{tab:nnz nns for decomposition} and Corollary \ref{cor: leading coefficient of bs}. The decomposed matrices are provided as supplemental material to this work \footnote{\supplementalurl}. The number of non zeros and the number of non singleton in the encoding and decoding matrices $U_\varphi, V_\varphi, W_\varphi$ are listed in Table \ref{tab:nnz nns for decomposition}. 

\begin{example}
    For $n_0=44$, we have $nnz(U_\varphi)=103661$, $nns(U_\varphi)=92$, $nnz(V_\varphi)=103822$, $nns(V_\varphi)=322$, $nnz(W_\varphi)=103753$, $nns(W_\varphi)=6532$, $t_0=36110$, $s_0=2116$ (see Table \ref{tab:nnz nns for decomposition}).

    We obtain the leading coefficient using the formula in Corollary \ref{cor: leading coefficient of bs}. Note that by Claim \ref{claim: qu qv qw number of linear operations},
    \begin{align*}
        q_{U_\varphi} &= nnz(q_{U_\varphi}) + nns(q_{U_\varphi}) - nrows(U_\varphi) = 103753 - t_0 = 67643 \\
        q_{V_\varphi} &= nnz(q_{V_\varphi}) + nns(q_{V_\varphi}) - nrows(U_\varphi) = 104144 - t_0 =  68034 \\
        q_{W_\varphi} &= nnz(q_{W_\varphi}) + nns(q_{W_\varphi}) - ncols(W_\varphi) = 110285 - s_0 = 108169
    \end{align*}
    and thus
    \begin{align*}
        c=\frac{q_{U_\varphi}+q_{V_\varphi}+q_{W_\varphi}}{t_0-s_0}+1=\frac{243846}{33994} + 1 \approx 8.174
    \end{align*}
\end{example}

\begin{table}
    \centering
    \begin{tabular}{|c|c|c|c|c|c|c|c|c|c|}
    \hline
    $n_0$ & $nnz(U_\varphi)$ & $nns(U_\varphi)$ & $nnz(V_\varphi)$ & $nns(V_\varphi)$ & $nnz(W_\varphi)$ & $nns(W_\varphi)$ & $t_0$ & $s_0$ & $c$\\
    \hline\hline
$20$ & $12089$ & $44$ & $12166$ & $154$ & $12133$ & $1540$ & $4378$ & $484$ & $8.419$ \\ 
\hline
$30$ & $35824$ & $64$ & $35936$ & $224$ & $35888$ & $3200$ & $12688$ & $1024$ & $8.265$ \\ 
\hline
$40$ & $79359$ & $84$ & $79506$ & $294$ & $79443$ & $5460$ & $27748$ & $1764$ & $8.193$ \\ 
\hline
$42$ & $90970$ & $88$ & $91124$ & $308$ & $91058$ & $5984$ & $31746$ & $1936$ & $8.183$ \\ 
\hline
$44$ & $103661$ & $92$ & $103822$ & $322$ & $103753$ & $6532$ & $36110$ & $2116$ & $8.174$ \\ 
\hline
$46$ & $117480$ & $96$ & $117648$ & $336$ & $117576$ & $7104$ & $40856$ & $2304$ & $8.165$ \\ 
\hline
$48$ & $132475$ & $100$ & $132650$ & $350$ & $132575$ & $7700$ & $46000$ & $2500$ & $8.158$ \\ 
\hline
$50$ & $148694$ & $104$ & $148876$ & $364$ & $148798$ & $8320$ & $51558$ & $2704$ & $8.151$ \\ 
\hline
$60$ & $249829$ & $124$ & $250046$ & $434$ & $249953$ & $11780$ & $86118$ & $3844$ & $8.124$ \\ 
\hline
    
    \end{tabular}
    \caption{Detailed analysis of our decomposed algorithms. The algorithm uses a fast basis transformation $\varphi:\mathbb{R}^{n_0^2}\to\mathbb{R}^{s_0}$, and encoding/decoding matrices $U_\varphi, V_\varphi, W_\varphi \in \mathbb{R}^{t_0\times s_0}$. The resulting algorithm has an additive complexity of $c\cdot n^{\log_{n_0}t_0} + o\left(n^{\log_{n_0}t_0}\right)$. Recall that $nnz$ is the number of elements in the matrix that are non-zeros, $nns$ is the number of elements in the matrix that are non-singletons, that is, not in $\left\{-1,0,1\right\}$, $U_\varphi, V_\varphi, W_\varphi$ are $t_0\times s_0$ matrices and $c$ is the resulting leading coefficient.}
    
    \label{tab:nnz nns for decomposition}
\end{table}

\section{Composition of Algorithms}
\label{app: composition of algorithms}
In this section we analyze the composition of our algorithms with other algorithms. Specifically, we show that composing our $\left<44,44,44;36110\right>$-algorithm with Strassen's $\left<2,2,2;7\right>$-algorithm will require less multiplication compared to our $\left<88,88,88;257100\right>$

\begin{claim}
    The algorithm $\szone_{88}$ is an $\left<88,88,88;257100\right>$-algorithms.
\end{claim}

\begin{proof}
    Substituting $n_0=88$ in Theorem \ref{theorem: we have 2227 alg with U0=KU and V0=KV} we get that $\szone_{88}$ requires $t=257100$ multiplications. 
\end{proof}

\begin{claim}
    Let $ALG$ be the algorithm obtained by composing $\szone_{44}$ with Strassen's $\left<2,2,2;7\right>$-algorithm. Then $ALG$ is an $\left<88,88,88,252770\right>$-algorithm.
\end{claim}

\begin{proof}
    The claim follows immediately from Claim \ref{claim: composition of algorithms}.
\end{proof}

\end{document}